\documentclass[review]{elsarticle}
\usepackage{graphicx}
\usepackage{amsfonts}
\usepackage{lipsum}
\usepackage{amsmath}
\usepackage{color}
\usepackage{enumerate}
\usepackage{lineno,hyperref}
\modulolinenumbers[5]
\usepackage{amssymb}
\usepackage{stix}

\usepackage{graphicx}
\usepackage{amsfonts}
\usepackage{amsmath}
\usepackage{color}
\usepackage{lineno,hyperref}
\modulolinenumbers[5]
\usepackage{amssymb}
\usepackage{amsthm}
\usepackage{ragged2e}
\usepackage{changepage}

\newtheorem{Assumption}{Assumption}

\newtheorem{Proposition}{Proposition}

\newtheorem{Problem}{Problem}

\newtheorem{Definition}{Definition}
\newtheorem{Theorem}{Theorem}

\DeclareMathOperator*{\diag}{diag}
\DeclareMathOperator{\Tr}{Tr}

\usepackage{calrsfs}
\DeclareMathAlphabet{\pazocal}{OMS}{zplm}{m}{n}

\usepackage{calrsfs}

\usepackage{mathtools}
\mathtoolsset{showonlyrefs}

\usepackage{calrsfs}
\DeclareMathAlphabet{\pazocal}{OMS}{zplm}{m}{n}

\usepackage{calrsfs}

\usepackage{mathtools}
\mathtoolsset{showonlyrefs}

\begin{document}

\begin{frontmatter}

\title{Mixed $H_2/H_{\infty}$ Control of Delayed Markov Jump Linear Systems}

\author[myfirstaddress,mysecondaddress]{Wenjie Mei\corref{mycorrespondingauthor}}
\cortext[mycorrespondingauthor]{Corresponding author.}
\ead{Wenjie.Mei@inria.fr}

\address[myfirstaddress]{Division of Information Science, Nara Institute of Science and Technology, Nara 630-0101, Japan}
\address[mysecondaddress]{Current affiliation: Inria, Univ. Lille, CNRS, UMR 9189 - CRIStAL,
F-59000 Lille, France}
\address[mythirdaddress]{Current affiliation: Graduate School of Information Science and Technology,
Osaka University, Osaka 565-0871, Japan}

\author[myfirstaddress]{Chengyan Zhao}
\author[myfirstaddress,mythirdaddress]{Masaki Ogura}

\author[myfirstaddress]{Kenji Sugimoto}

\begin{abstract}
This paper investigates state feedback control laws for Markov jump linear systems with state and mode-observation delays. An assumption in this study is that the delay of mode observation obeys an exponential distribution. Also, we raise an unknown time-varying state delay applied in the composition of the state feedback controller. A method of remodeling the closed-loop system as a standard Markov jump linear system with state delay is shown. Furthermore, on the basis of this remodeling, several Linear Matrix Inequalities (LMI) for designing feedback gains for stabilization and mixed $H_2/H_{\infty}$ control are proposed. Finally, we apply a numerical simulation for examining the effectiveness of the proposed mixed $H_2/H_{\infty}$ controller designing method. 
\end{abstract}


\end{frontmatter}

\section{Introduction}\label{sec1}

Markov jump linear systems are an essential subclass of switched linear systems and are customarily applied to model the system with sudden variations in their constructions, which in part results from the inherent vulnerability of dynamic systems to component failures, unanticipated environmental interference, alternation of subsystems' interconnection, and unexpected variations in the
operations of plants~[1]. 

Time delays appear in many physical processes and usually cause low performance, oscillation, and instability.  Thus, the Markov jump linear systems with time delays have received extensive attention. The existing studies are mainly of two classes: delay-independent and delay-dependent approaches. The investigations~[2--4] showed that delay-dependent consequences are usually stricter than the ones in the delay-independent case, particularly if the delay's range is not large. The majority of existing investigations share a hypothesis that the controller detects the state and the current system mode simultaneously. Nevertheless, it unavoidably costs time for identifying the mode of systems and the state, and then switches to the target controller in reality. In~[5] some chemical systems illustrate the necessity of the research of asynchronous switching for control with high efficiency. The work in~[6--9] investigated the principal control method to let an unstable system keep stable with or without time delays, i.e., set up appropriate input signals to achieve the above goal, which requires knowledge of the current system mode, state or output. However, the methods would be failed since the difficulty to receive the information of the current system mode, state, and output. To solve the issue, Xiong et al.~[10] designed time-delay controllers by the past system information, rather than the current system information. Another important type of delay is the delay of the observation process when observing the switching of the mode signal. In this case, the observation signal is usually assumed to follow a regulation, e. g., a renewal process. Cetinkaya et al.~[11] devised almost-surely state feedback controllers for stabilization with resetting of gains whenever an observation occurs. The authors in~[12, 13] devised state-feedback controllers for stabilization with periodical observations. From above, we see that there is usually no efficient method to obtain the immediate information of switching signal, thus in this paper, we assume that there exists mode-observation delay following an exponential distribution, which occurs naturally when describing the time for a continuous process to change state when measuring the signal.

Delay-dependent stochastic stability is widely used in stability analysis of Markovian jump systems with delay. It is a fundamental property of state stability and is much more complex than in the deterministic case. It also presents some interesting performance and has been studied in many practical fields, e.g., Hopfield neural networks~\textcolor{green}{[14]}, uncertain neutral stochastic systems~\textcolor{green}{[15]}. Some studies on feedback control of Markov jump systems with delay were presented in such as~\textcolor{green}{[16,17,18,19,20]}. Several existing studies of $H_2/H_{\infty}$ control  with feedback law of switched systems including Markov jump systems are: Guaranteed cost control of uncertain periodic piecewise linear systems~\textcolor{green}{[21]}; $H_2/H_{\infty}$ control of switched TS fuzzy systems with time-delay~\textcolor{green}{[22]}; mixed $H_2/H_{\infty}$ control of stochastic systems~\textcolor{green}{[23]}; $H_{\infty}$ control for Markov jump systems with bounded transition probabilities~\textcolor{green}{[24]} and uncertain transition probabilities~\textcolor{green}{[25]}, and subject to actuator saturation~\textcolor{green}{[26]}, to mention a few results.

Although there exist numerous studies on $H_2/H_{\infty}$ control of Markov jump systems, there is still a lack of research on mixed $H_2/H_{\infty}$ control, as well as the mode-observation delay. Compared with the existing works, the main advantage of this study is that based on the assumption on the mode-observation with reflecting the general occurrence of observation delay and its properties, it then becomes possible to use existing methodologies of standard delayed Markov jump systems to remodel the system with state and mode-observation delays since the extended state remains a Markov process under the assumption. Without the assumption, the Markov jump linear systems with state and mode-observation delays would not be in a standard form, which increases the difficulty of studying their stability and control. A novel framework for studying state feedback control laws for continuous-time Markov jump linear systems with state and mode-observation delays is proposed in the paper. Furthermore, for realizing stabilization and mixed $H_2/H_{\infty}$ control we raise a set of LMIs to achieve state feedback control. 

The remaining of this paper is organized as follows: The problem of state feedback control of Markov jump linear systems with state and mode-observation delays is presented in Section~\ref{probelm_formulate}. We study that it is possible to transform the closed-loop system into a Markov jump linear system in a standard form by integrating the mode signal and the observation signal in the closed-loop system into a Markov process in Section~\ref{reduction}. In Section~\ref{main_result},  for stabilization and mixed $H_2/H_{\infty}$ control, we formulate an LMI framework for designing a state feedback controller. Section~\ref{numerical_example} shows a numerical example.

\subsection*{Notation}
This paper's notation is standard. $\mathbb{Z}_{+}$, $\mathbb{R}$, and $\mathbb{R_{+}}$ represent the sets of positive integers, real numbers, and positive real numbers. For $n,m \in \mathbb{Z}_{+}$, we use $\mathbb{R}^n$ and $\mathbb{R}^{n \times m}$ to indicate the spaces of $n$-dimensional real vectors and $n \times m$ real matrices. The Euclidean norm on $\mathbb{R}^n$, the probability of an event, and the mathematical expectations are denoted by $\rVert \cdot \rVert$, ${\rm{Pr}}(\cdot)$, and $E[\cdot]$, respectively. Given any real matrix $A$, we use $A^\top$ and $\Tr (A)$ to indicate the transpose and the trace of matrix $A$, respectively. We let $A>0$ represent that $A$ is positive-definite. When $A$ is symmetric, the maximum and minimal eigenvalue of $A$ are denoted by $\lambda_{\max}(A)$ and $\lambda_{\min}(A)$, respectively. We let $\diag(\cdot)$ represent a block diagonal matrix. Let $I_{m}$ denote the $m \times m$ identify matrix. 
Let $\bf{1}$ denote indicator functions. In this paper, we use $\star$ to denote the symmetric blocks of partitioned symmetric matrices. Given $q \in \mathbb{Z}_{+}$, we denote by $\mathcal{L}_2([0,\infty),\mathbb{R}^q)$ the $q$-dimensional vector space of all functions $f \colon [0,\infty) \to \mathbb{R}$ satisfying  $\int_0^{\infty} \rVert f(t) \rVert^2 dt < \infty$. Given $\tau \in \mathbb{R}_{+}$, we let $\mathcal{C}([-\tau,0] \rightarrow \mathbb{R}^n)$ denote the space of functions mapping $[-\tau,0]$ to $\mathbb{R}^n$.

\section{Problem Statement}
\label{probelm_formulate}

We introduce the considered problems in this paper in this section. Let $r = \{r(t)\}_{t\geq 0}$ be a continuous-time Markov process taking values in $\Theta=\{1,2,\dotsc,N\}$ and having the following transition probabilities for all $h>0$ and  $i, j \in \Theta$:
\begin{equation} \label{markov_ process_r}
{\rm{Pr}}(r(t+h)=j \mid r(t)=i) = 
\left \{  
\begin{aligned}
&  \pi_{ij} h +  \mathcal{O}(h), & \mbox{if} \ i \neq j, &  \\ 
&  1  +\pi_{ii}h +\mathcal{O}(h), & \mbox{if} \  i=j, &    
\end{aligned}  
\right.
\end{equation}
where $\lim_{h \to \infty} {\mathcal{O}(h)}/h = 0$ and $\sum_{j \in \Theta \backslash \{i\} }\pi_{ij}=-\pi_{ii}$. The initial condition of $r$ is denoted by $r(0) =r_0$.
For each $i\in \Theta$, we let $A_{i} \in \mathbb{R}^{n \times n}, B_{i} \in \mathbb{R}^{n \times m}, C_{i} \in \mathbb{R}^{l \times n}, J_{i} \in \mathbb{R}^{l \times n}, E_{i} \in \mathbb{R}^{n \times q}, \varPsi_{i} \in  \mathbb{R}^{l \times q}$, and $\varPhi_{i} \in  \mathbb{R}^{l \times q}$ be matrices. We consider the Markov jump linear system
\begin{equation} \label{mixed H control}
\Sigma \colon \left\{ 
\begin{aligned}
\dot x (t) &= A_{r(t)}x(t)+B_{r(t)}u(t) + E_{r(t)} w(t), \\ 
z(t) &= C_{r(t)}x(t) +   \varPsi_{r(t)}w(t), \\
y(t) &= J_{r(t)}x(t) +\varPhi_{r(t)}w(t), 
\end{aligned}
\right.
\end{equation}
where $x \in \mathbb{R}^n$ is the state vector, $u \in \mathbb{R}^m$ is the input, $w \in \mathcal{L}_2([0,\infty),\mathbb{R}^q)$ is a disturbance, $z$ is the controlled output, while $y$ is the measurable output. 


\subsection{State and mode observation delays}

We investigate the problem of mixed $H_2/H_{\infty}$ control of the system~$\Sigma$ by applying delayed state-feedback with mode-dependent feedback gains in this paper. We specifically allow delays in the measurement of both $x$ and~$r$. Therefore, we introduce the state feedback control
\begin{equation} \label{controller}
u(t)=K_{\tilde r(t)} x(t-\tau(t)), 
\end{equation}
where $K_1, \dotsc, K_N \in \mathbb{R}^{m \times n}$ are state-feedback gains  and $\tau$ represents an unknown time-varying delay with 
\begin{equation*}
\mbox{$\tau(t) \in [0, \tau_0]$,\   \ $\dot{\tau}(t) \in [0,\tau^{+}]$}
\end{equation*} 
for some positive constants $\tau_0$ and $\tau^{+} \in (0, 1]$, while $\tilde r$ represents a delayed measurement of the mode signal that is described as follows: 
\begin{enumerate}
	\item When the mode signal $r$ changes from $i$ to $j \in \Theta \backslash \{i\}$ at time $t$, a constant $h$ is drawn from a fixed probability distribution. We call the random variable~$h$ the mode observation delay. 
	\item If $r$ remains to be $j$ until time $t+h$, then the value of $\tilde r$ is set to $j$ at time $t+h$. 
	\item On the other hand, if $r$ changes its state before time $t+h$, then we go back to the first step.
\end{enumerate}
These properties are illustrated by Fig.~\ref{fig:random}. 

\begin{figure}[tb] 
	\centering{\includegraphics[width=0.35\textwidth]{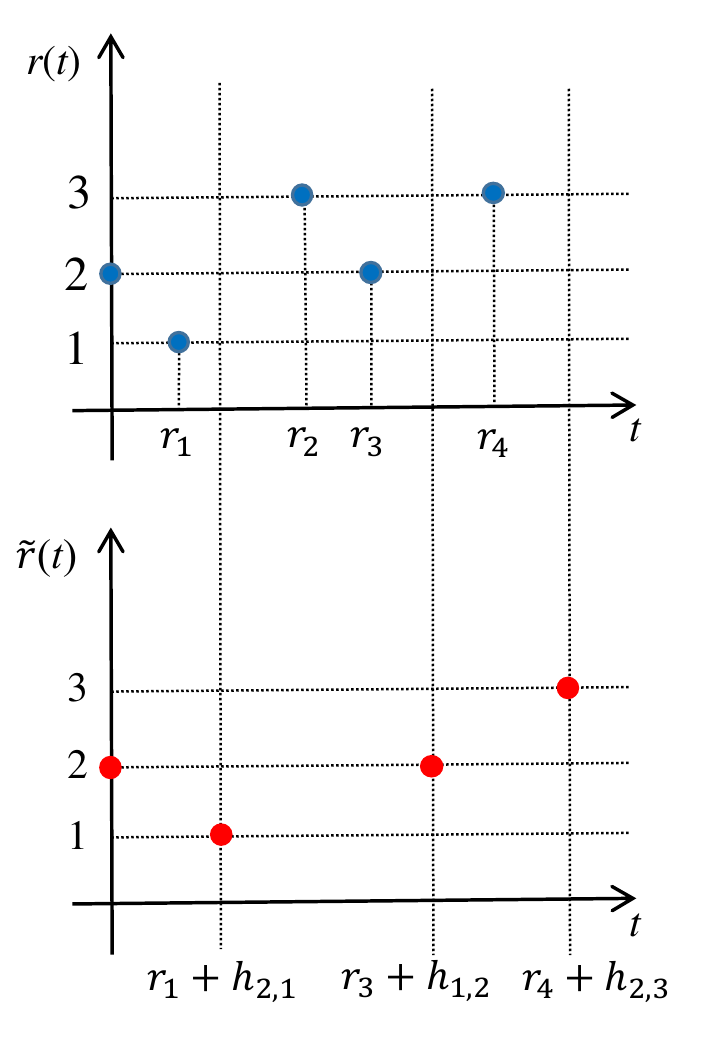}}
	\caption{{An observation of the mode signal with $\Theta = \{1,2,3\} $. Let $r_{\sigma}$ represent the $\sigma$-th switching of the process $r$ and $h_{j_1,j_2}$ represent the mode observation delay starting from the most recent $r$ of $\tilde{r} \coloneqq \{\tilde r(t)\}_{t \geq 0}$ switching from the state $j_1$ to the state $j_2$ as shown in this figure, where $\sigma \in \mathbb{Z}_{+}$, $j_1, j_2 \in \Theta$. Until the first observation $r_1 + h_{2,1}$, we let $\tilde r$ be set to $2$. For example, $r_1+h_{2,1}$ represents the mode observation time from $r_1$, in which $h_{2,1}$ represents a mode observation delay of $\tilde r$ jumping from the state $2$ to the state $1$. In a similar way, $h_{1,2}$ denotes a mode observation delay under the switching of $\tilde r$ from the state $1$ to the state $2$.}
\label{fig:random}}
\end{figure}

We place the assumption on the mode observation delay in this paper as follows.

\begin{Assumption}\label{asm:}
The mode observation delay $h_{i,j}$ from the current observation state $i$ to the correct mode $j$ follows an exponential distribution with rate $g_{ij}> 0$ for each $i, j \in \Theta$. 
\end{Assumption} 

\subsection{Problem formulation}

The system~$\Sigma$ and the state feedback control~\eqref{controller} yield the following closed-loop system 
\begin{equation}  \label{stabilized_Markov jump linear system}
\Sigma_K \colon \left\{ 
\begin{aligned} 
\dot x (t) &= A_{r(t)}x(t) + B_{r(t)}K_{\tilde r(t)} x(t-\tau(t)) +E_{r(t)} w(t), \\ 
&x(\varepsilon) = \phi(\varepsilon), \varepsilon \in [-\tau_0,0], \\
z(t) &= C_{r(t)}x(t) +  \varPsi_{r(t)}w(t), \\
y(t) &= J_{r(t)}x(t) +  \varPhi_{r(t)}w(t), \\
\end{aligned}
\right.
\end{equation}
where $\phi \in \mathcal{C}([-\tau_0,0] \rightarrow \mathbb{R}^n)$ is the initial state. The weak delay-dependent stochastic stability for the system~$\Sigma_K$ is defined as follows. 

\begin{Definition}
If there exists a real number $C((r_0, \tilde r_0), \phi) > 0$ for every $((r_0, \tilde r_0), \phi)$ such that 
	\begin{equation}
	\begin{aligned} 
	 E \bigg[\int_0^{\infty} \rVert x(t,\phi) \rVert^2 dt  \mid (r_0, \tilde r_0), & \phi(\varepsilon), \varepsilon \in [-\tau_0,0] \bigg] \\  
	&\leq C((r_0, \tilde r_0), \phi) 
	\end{aligned}
	\end{equation}
    for all $w \in \mathcal{L}_2([0,\infty),\mathbb{R}^q)$, then~$\Sigma_K$ is said to be \emph{weakly delay-dependent stochastically stable}. 
\end{Definition}

We then introduce the $H_2$ and $H_\infty$ performance measures of~$\Sigma_K$. 

\begin{Definition}
	Let a constant $\gamma>0$ be given. Consider the system~$\Sigma_K$, we define the $H_2$ performance measure as 
	\begin{equation} \label{h2 performance}
	\mathcal{H}_2 = E \bigg[ \int_0^{\infty} z^\top (t) z(t) dt \bigg], 
	\end{equation}
	and $H_{\infty}$ performance measure as
	\begin{equation}\label{hinf performance}
	\mathcal{H}_{\infty} = E \bigg[ \int_0^{\infty} \Big( - \gamma^2 w^\top(t) w(t) + y^\top (t) y(t) \Big) dt \bigg], 
	\end{equation}
	respectively. 
\end{Definition}

We now illustrate the problem that we study in this paper. 

\begin{Problem} 
Let $f_2>0$, $f_{\infty} >0$, and $\gamma >0$ be given. Find state-feedback gains $K_1$, \dots, $K_N$ such that the  system~$\Sigma_K$ is weakly delay-dependent stochastically stable and 
	\begin{equation} \label{h_2 and h_inf measure}
	\mathcal{H}_2 \leq f_2 \ \mbox{and} \ \sup_{w \in \mathcal{L}_2([0,\infty),\mathbb{R}^q)} \mathcal{H}_{\infty} \leq  f_{\infty}
	\end{equation}
	are satisfied. 
\end{Problem}

We remark that~$\Sigma_K$ is not a delayed Markov jump linear system in a standard form since there is a random mode delay in the observation signal. For this reason, we cannot apply the methodologies of controlling standard delayed Markov jump linear systems available in the literature. 
To address this issue, we show that the system~$\Sigma_K$ can be remodeled as a standard Markov jump linear system with state delay in the next section.

\section{Equivalent Reduction} \label{reduction}

We remodel~the system $\Sigma_K$ as a delayed Markov jump linear system in a standard form by embedding $r$ and $\tilde r$ in the system $\Sigma_K$ to a Markov process in this section. 
%
%
%
%
	For the system~$\Sigma_K$, we introduce the Markov process 
	\begin{equation}
	s(t) = (r(t),\tilde r(t)),
	\end{equation}
	which gathers the continuous-valued stochastic processes $r$ and $\tilde r$. Besides, let $\mathcal{D}$ denote the set of $\xi= (i,j) \in \Theta \times \Theta$. The set $\mathcal{D}$ has all probable values taken by $s = \{s(t)\}_{t \geq 0}$. The initial condition $s(0)$ of $s$ is denoted by $s_0$.

Let us then show the following proposition, which allows us to rewrite $\Sigma_K$ as a standard delayed Markov jump linear system. The proof for Proposition~\ref{reduction_MJLS} is omitted because it is a direct result considering the definition of the observation process $\tilde r$ as well as Assumption~\ref{asm:}.

\begin{Proposition} \label{reduction_MJLS}
The stochastic process $s$ is a time-homogeneous Markov process taking values in $\mathcal{D}$. Also, the transition rate from $\xi = (i_1,j_1) \in \mathcal{D}$ to $\xi' = (i_2,j_2) \in \mathcal{D}$ is given by
	\begin{alignat}{3}
	& q_{\xi \xi'}   = 
	\left \{  
	\begin{aligned}
	&  {\bf{1}}(j_1 = j_2) \pi_{i_1i_2} ,  && \mbox{if no observation} \\ 
	&  {\bf{1}}(i_1 = i_2 = j_2) g_{j_1 j_2, }, && \mbox{otherwise}.    
	\end{aligned}  
	\right.  
	\end{alignat}
	
\end{Proposition}

For a family of matrices $\{D_1, \dotsc, D_N\}$, we use the notations
\begin{equation}
\hat D_{i,j} = D_i
\end{equation}
and 
\begin{equation}
\check D_{i,j} = D_j
\end{equation}
for all $i, j \in \Theta$. Then, by Proposition~\ref{reduction_MJLS}, we see that the system~$\Sigma_K$ can be rewritten to the following standard delayed Markov jump linear system:
\begin{equation}  
\bar \Sigma_K \colon \left\{ 
\begin{aligned} 
\dot x (t) &= \hat A_{s(t)}x(t)+ \hat B_{s(t)} \check K_{s(t)} x(t-\tau(t)) +\hat E_{s(t)} w(t), \\ &x(\varepsilon) = \phi(\varepsilon), \varepsilon \in [-\tau_0,0], \\
z(t) &= \hat C_{s(t)}x(t) +  \hat \varPsi_{s(t)}w(t), \\
y(t) &= \hat J_{s(t)}x(t) + \hat \varPhi_{s(t)}w(t). \\
\end{aligned}
\right.
\end{equation}

\section{Main result} \label{main_result}

The main result of devising a mixed $H_2/H_{\infty}$ controller is shown in this section. 
Throughout this paper, we fix a one-to-one mapping $\mathcal F \colon \Theta\times \Theta \mapsto \{1, \dotsc, N^2\}$. In the sequel, we use the notation 
\begin{equation} 
k_{ij} = \mathcal{F}((i,j)),\ \mathcal{F}_0 = \mathcal{F}(s_0).
\end{equation}
We also let the infinitesimal generator of the stochastic process $s$ be denoted by $\tilde{\mathcal{S}} = [\kappa_{kk'}]_{kk'}\in \mathbb{R}^{N^2\times N^2}$. 
The following theorem is for devising a state feedback controller having a prescribed $H_2$ and $H_\infty$ performance measure and is the main result in our study.

\begin{Theorem}\label{thm:main:control}
	Let $f_2 > 0, f_{\infty} > 0, \gamma>0$ be given constants and $L_k \in \mathbb{R}^{n \times n}>0$ be given matrices, respectively. For the system~$\Sigma_K$ the feedback gain $$K_j = Z_j Y^{-1}_j \in \mathbb{R}^{m \times n}$$ is a mixed  $H_2/H_{\infty}$ controller satisfying the limits of $H_2$ and $H_{\infty}$ performance measures in~\eqref{h_2 and h_inf measure} if there exist matrices $Y_j = Y_j^\top \in \mathbb{R}^{n \times n}$, $\lambda \in \mathbb{R}$, $\Lambda \in \mathbb{R}$, and $Z_j \in \mathbb{R}^{m \times n}$
satisfying the following system of LMIs: 
\\
	\begin{alignat}{3} \label{combineLMI}
& \lambda + \Lambda \leq \min\{f_2,f_{\infty}\};
\\
	& \begin{bmatrix} 
	\mathcal{U}_{k_{ij}} & \star & \star & \star & \star   \\
	Y_j & -L_{k_{ij}} & \star & \star & \star  \\
	C_iY_j & 0 & -I_l & \star & \star  \\
	Z_j^\top B_i^\top & 0 & 0 & I_n & \star \\
	\aleph_{k_{ij}}^\top & 0 & 0 & 0 & -\mathcal{Y}_{k_{ij}} \\
	\end{bmatrix} < 0; \\
	&\begin{bmatrix} 
	\mathcal{U}_{k_{ij}} & Y_j & Y_j J_i^\top & B_iZ_j & \aleph_{k_{ij}}  & E_i + Y_j J_i^\top \varPhi_i \\
	\star & -L_{k_{ij}} & 0 & 0 & 0 & 0 \\
	\star & \star & -I_l & 0 & 0 & 0 \\
	\star & \star & \star & I_n & 0 & 0\\
    \star & \star & \star & \star & -\mathcal{Y}_{k_{ij}} & 0\\
	\star & \star & \star & \star & \star & -\gamma^2I + \varPhi_i^\top \varPhi_i\\
	\end{bmatrix} < 0; \\
	& \begin{bmatrix}
	-\lambda  & \star  \\
	\phi(0) & -Y_j
	\end{bmatrix} \leq 0, 
	\begin{bmatrix}
	-\Lambda & \star  \\
	X & -\frac{1}{\lambda_{\max}(L_k^{-1})}
	\end{bmatrix}  \leq  0;  \\
	& Y_j>0; \; \lambda >0; \; \Lambda > 0
	\end{alignat}
	for all $k_{ij} \in \Theta \times \Theta, i, j \in \Theta$, where the matrices $\mathcal{U}_k$, $\mathcal{Y}_k$, $\aleph_k$, and $X$ are defined by
 \begin{alignat}{3}
	&\mathcal{U}_{k_{ij}} = Y_j A_i^\top + A_i Y_j + \kappa_{{k_{ij}}{k_{ij}}} Y_j \in \mathbb{R}^{n \times n}, \\  
	&\mathcal{Y}_{k_{ij}} = \diag(Y_1, \dots,Y_{{k_{ij}}-1}, Y_{{k_{ij}}+1},\dots, Y_{N^2}) \\ 
	& \quad \quad \quad \in \mathbb{R}^{(N^2-1) \times (N^2-1)}, \\
	&\aleph_{k_{ij}} = [\sqrt{\kappa_{{k_{ij}}1}}Y_j \dotsc \sqrt{\kappa_{{k_{ij}}({k_{ij}}-1)}}Y_j, \sqrt{\kappa_{{k_{ij}}({k_{ij}}+1)}}Y_j \dotsc \\ 
	&  \quad \quad \quad \sqrt{\kappa_{{k_{ij}} N^2}}Y_j] \in \mathbb{R}^{n \times (N^2-1)}, \\
	&X = \sqrt{\bigg( \int_{-\tau}^0 x^\top(t)x(t) dt \bigg)} \in \mathbb{R}.
   \end{alignat}
\end{Theorem}

We remark that both the second and third LMIs in~\eqref{combineLMI} depend on $g_{ij}$ since the value of $g_{ij}$ affects $\kappa_{ij}$, which is embedded in the LMIs.

\subsection{Proof}
The proof of Theorem~\ref{thm:main:control} is presented in this subsection. The two propositions in the sequel are proposed for devising a mixed $H_2/H_{\infty}$ controller for the system~$\Sigma_K$, in which the first proposition provides the conditions of a $H_2$ controller with limited $H_2$ performance measure.

\begin{Proposition} \label{H_2 controller theorem}

Let $Q_k \in \mathbb{R}^{n \times n}$  be positive definite matrices. Assume that there exist positive definite matrices $P_k \in \mathbb{R}^{n \times n}$ such that
	\begin{equation} \label{varXi}
	\begin{aligned}
	\varXi_{k_{ij}} = 
	\centering{\begin{matrix}
		\begin{bmatrix}
		\eth_{k_{ij}} & \star   \\
		K_j^\top B_i^\top P_{k_{ij}} & -\hat Q_{k_{ij}}
		\end{bmatrix} 
		\end{matrix} }  < 0
	\end{aligned}
	\end{equation}	
	are satisfied for all $(i,j) \in \Theta \times \Theta$, where 
\begin{equation}
\begin{aligned}
\eth_{k_{ij}} = A_i^\top P_{k_{ij}} &+  P_{k_{ij}} A_i  + Q_{k_{ij}} \\ &+\sum_{k'=1}^{N^2} \kappa_{{k_{ij}}k'}P_{k'} + C^\top_i C_i 
\in \mathbb{R}^{n \times n}
\end{aligned}
\end{equation}
and $\hat Q_{k_{ij}} = (1-\tau^{+})Q_{k_{ij}}$. If $w(t) \equiv 0$, then 
	\begin{equation} \label{h_2_bound}
	\mathcal{H}_2 \leq  x^\top_0P_{\mathcal{F}_0}x_0 + \int_{-\tau}^0 \phi^\top(v) Q_{\mathcal{F}_0} \phi(v)dv.
	\end{equation}
\end{Proposition}

\begin{proof}
	Consider a Lyapunov function 
	\begin{equation} \label{function V}
	\begin{aligned}
	V(x(t),t,k) &= x^\top(t)P_kx(t) + \int_{t-\tau}^t x^\top(v) Q_k x(v) d v, \\ k &= {\mathcal{F}(s(t))}.  
	\end{aligned}
	\end{equation}
	By~[27, 28], we have that the weak infinitesimal operator $\mathcal{S}^x[\cdot]$ of $\{x(t)\}_{t \geq 0}$ in $ \Sigma_K$ is 
	\begin{equation}
	\mathcal{S}^x[V] = \frac{\partial V}{  \partial t} + \dot{x}^\top(t) \frac{\partial V}{\partial x} \big|_{k} + \sum_{k'=1}^{N^2} \kappa_{kk'}V.
	\end{equation}
	We let 
	\begin{equation}
	\begin{aligned}
	\widetilde{\varXi}_{{k_{ij}}} = A_i^\top P_{k_{ij}} + P_{k_{ij}} A_i &+ Q_{k_{ij}} +\sum_{k'=1}^{N^2} \kappa_{{k_{ij}}k'}P_{k'} \\ &+P_{k_{ij}} B_i K_j \hat Q^{-1}_{k_{ij}} K_j^\top B_i^\top P_{k_{ij}}. 
	\end{aligned}
	\end{equation}
	Let
	\begin{equation}
	\begin{aligned}
	\bar{\varXi}_{{k_{ij}}} = \widetilde{\varXi}_{{k_{ij}}} + C^\top_i C_i
	\end{aligned}
	\end{equation}
    for each $k_{ij}$, we see that $\bar{\varXi}_{{k_{ij}}}<0$ is equivalent to~\eqref{varXi}.
	By the proof development of Theorem~3.1 of~[29], the inequality
	\begin{equation} \label{S[V] ineq}
	\mathcal{S}^x[V] \leq x^\top(t) \widetilde{\varXi}_{k} x(t)
	\end{equation}
	holds. Besides, the fact $C^\top_i C_i>0$ and the LMIs~\eqref{varXi} lead to \begin{equation} \label{tilde_ineq}
	\widetilde{\varXi}_{k} < 0
    \end{equation}
	for all $i,j \in \Theta$.
	
	Since $\psi \rVert x(t) \rVert^2 \geq \rVert x(t+\vartheta) \rVert^2$ for some $\psi \in \mathbb{R}_{+}$ and all $ -\tau \leq \vartheta \leq 0$, it holds that
	\begin{equation} \label{V(x)_ineq}
	 x^\top(t)P_kx(t) + \mu \rVert x \rVert^2 \geq V(x(t),t,k) 
	\end{equation}
	by~\eqref{function V}, where $\mu = \psi \tau \lambda_{\max}(Q_k)$. The equations~\eqref{S[V] ineq},~\eqref{tilde_ineq}, and~\eqref{V(x)_ineq} show that
	\begin{equation}
	\begin{aligned}
	\frac{\mathcal{S}^x[V]}{V(x(t),t,k)} & \leq -\min_{k \in \{1,\dotsc,N^2\}} \bigg \{\frac{\lambda_{\min}(-\widetilde{\varXi}_k)}{\lambda_{\max}(P_k) + \mu} \bigg \} \\ &= -\zeta < 0,
	\end{aligned}
	\end{equation} 
    which induces 
	\begin{equation} \label{inqu_s_zeta}
	\mathcal{S}^x[V] \leq -\zeta V(x(t),t,k).
	\end{equation}
	In accordance with the formula
	\begin{equation} \label{Dynkin's formula}
	\begin{aligned}
	E \bigg[\int_0^t  \mathcal{S}^x[&V(x(v),v,\mathcal{F}(s(v)))] dv \bigg] \\
	&=-V(x_0,0,\mathcal{F}_0)+E[V(x(t),t,k)],
	\end{aligned}
	\end{equation}
	we obtain
	\begin{equation} 
	-\zeta V(x(t),t,k) \geq  E \big [ \mathcal{S}^x[V] \big] = \frac{dE[V(x(t),t,k)]}{dt},
	\end{equation}
	which results in
	\begin{equation}\label{expectation V}
	 e^{-\zeta t} V(x_0,0,\mathcal{F}_0) \geq E[V(x(t),t,k)].
	\end{equation}
	Consider~\eqref{inqu_s_zeta},~\eqref{Dynkin's formula}, and \eqref{expectation V}, we see that  
	\[
	\begin{aligned}
	&E \bigg[ \int_0^T  V(x(v),t,\mathcal{F}(s(v)))dv \bigg] \\
	&\leq -\frac{1}{\zeta} \Bigg[ \sup_{T \in [0.\infty)} \Big[ E\big[V(x(T),T,\mathcal{F}(s(T))) \big] -V(x_0,0,\mathcal{F}_0) \Big] \Bigg] \\
    &= -\frac{e^{-\zeta T}-1}{\zeta}	V(x_0,0,\mathcal{F}_0) \\
    &< \frac{1}{\zeta} V(x_0,0,\mathcal{F}_0)
	\end{aligned} 
    \]
	should be established, where $e^{-\zeta T}-1 \in (-1,0]$.
	
	It is a direct result that the inequality $E [\int_{-\tau}^0 x^\top(t+\sigma)Q_k x(t+\sigma) d\sigma] \geq 0$ holds. Therefore, 
	\begin{equation} 
	\begin{aligned}
	&\lim_{T \to \infty} E \bigg[\int_0^{T} x^\top(t) x(t) dt \mid (x_0,s_0)\bigg] \\
	& \leq  \lim_{T \to \infty} E \bigg[ \int_0^T  V(x(t),t,\mathcal{F}(s(t)))dt \mid (x_0,s_0)  \bigg]      \\
	&\leq x^\top_0 \lambda_{\max} ({\Omega}_k) x_0 + \frac{\tau_0 \lambda_{\max}(Q_{\mathcal{F}_0}) \big[\sup \rVert x(t+ \vartheta) \rVert^2 \big]}{\zeta \lambda_{\max} (P_k)}\\
	&< \infty, 
	\end{aligned}
	\end{equation}
	where 
	\begin{equation} 
	\begin{aligned}
	\Omega_k 
	= \max_{k \in \{1,\dotsc,N^2\}} \bigg\{ \frac{ P_{\mathcal{F}_0} }{ \zeta  \lambda_{\max} (P_k) } \bigg\}. 
	\end{aligned}
	\end{equation}
	This illustrates that the system~$\Sigma_K$ is weakly delay-dependent stochastically stable. Now, let us show the $H_2$ performance measure is limited by a level. Consider~\eqref{S[V] ineq} and~\eqref{Dynkin's formula}, we obtain
	\begin{equation}
	\begin{aligned}
	&E \big[V \big(x(\varGamma),\varGamma, \mathcal{F}(s(\varGamma)) \big) \big]-V(x_0,0,\mathcal{F}_0) \\ &\leq E \bigg[\int_0^{\varGamma} x^\top(v) \Big(\bar \varXi_{\mathcal{F}(s(v))} - C_{r(v)}^\top C_{r(v)} \Big) x(v) dv \bigg].
	\end{aligned}
	\end{equation}
	Then, it can be shown that
	\begin{equation} \label{h_2_theorem}
	\begin{aligned}
	\mathcal{H}_2 &= E \bigg[ \int_0^{\infty} z^\top(v)z(v)dv \bigg] \\ 
	&= E \bigg[ \int_0^{\infty} \big[x^\top(v) C_{r(v)}^\top  C_{r(v)} x(v) \big] dv \bigg] \\ 
	&\leq V(x_0,0,\mathcal{F}_0) \\
	&=  x^\top_0P_{\mathcal{F}_0}x_0 + \int_{-\tau}^0 \phi^\top(v) Q_{\mathcal{F}_0} \phi(v)dv
	\end{aligned}
	\end{equation}
	holds if~\eqref{varXi} holds true with $w(t) \equiv 0$ for all $t \geq 0$.
\end{proof} 

The following proposition is for the LMIs used to design a $H_{\infty}$ controller 
with the guarantee of $\gamma$ attenuation property and weak delay-dependent stochastic stability of the system~$\Sigma_K$, as well as a limitation of the defined $H_{\infty}$ performance measure.

\begin{Proposition} \label{H_f controller theorem}
Let $\gamma > 0$ be a given constant and $Q_k \in \mathbb{R}^{n \times n}$ be positive definite matrices. Assume that 
there exist positive definite matrices $P_k \in \mathbb{R}^{n \times n}$ such that the following LMIs 
	\begin{equation} \label{LMI_h-infty}
	\begin{aligned}
	\hat \varXi_{k_{ij}} 
	= 
	    \centering{\begin{matrix}
		\begin{bmatrix}
		\Re_{k_{ij}} &\star  &\star \\
		K_j^\top B_i^\top P_{k_{ij}} & -\hat Q_{k_{ij}} & \star \\
		E_i^\top P_{k_{ij}} + \varPhi^\top_i J_i  & 0 & -\gamma^2 I_q + \varPhi^\top_i \varPhi_i 
		\end{bmatrix} 
		\end{matrix}} 
	    < 0
	\end{aligned}
	\end{equation}
	are satisfied for all $(i,j) \in \Theta \times \Theta$, where the matrix $\Re_{k_{ij}}$ is defined by
\begin{equation}
\begin{aligned}
\Re_{k_{ij}} =  A_i^\top P_{k_{ij}} &+  P_{k_{ij}} A_i  + Q_{k_{ij}} \\&+\sum_{k'=1}^{N^2} \kappa_{{k_{ij}}k'} P_{k'} + J^\top_i J_i \in \mathbb{R}^{n \times n}.
\end{aligned}
\end{equation}
Then, the system~$ \Sigma_K$ is weakly delay-dependent stochastically stable and 
	\begin{equation} \label{h_f_boudn}
	\mathcal{H}_{\infty} \leq   x^\top_0 P_{\mathcal{F}_0}x_0 + \int_{-\tau}^0 \phi^\top(v) Q_{\mathcal{F}_0} \phi(v)dv 
	\end{equation}
is satisfied for all $w \in \mathcal{L}_2([0,\infty),\mathbb{R}^q)$. 
\end{Proposition}
\begin{proof}
	The proof of weak delay-dependent stochastic stability with $w(t) \equiv 0$ has been shown in the proof development of Proposition~\ref{H_2 controller theorem}.  Then let us prove that~$\Sigma_K$ has the disturbance attenuation level $\gamma$. 
	
	We have already considered the Lyapunov function $V$ in the proof of Proposition~\ref{H_2 controller theorem}. Then, let \[\mathcal{S}_x[V] = \mathcal{S}^x[V] + w^\top(t) E^\top_i P_k x(t) + x^\top(t) P_k E_i w(t) .\] 
	Applying Dynkin's formula~\eqref{Dynkin's formula}, we can show that 
	\begin{equation} \label{h_infy_ineq}
	\mathcal{X}^\top(t) \hat \varXi_k \mathcal{X}(t) \geq - \gamma^2 w^\top(t) w(t) +y^\top(t)y(t) + \mathcal{S}_x[V] 
	\end{equation}
	by $\Sigma_K$ and \eqref{S[V] ineq}, where 
	\begin{equation}
	\mathcal{X}(t) = \begin{bmatrix}
		x(t) \\
		x(t-\tau)  \\
		w(t)
	\end{bmatrix}.
	\end{equation}
	The inequality~\eqref{h_infy_ineq} can be transformed into 
	\begin{equation} \label{H_T}
	\begin{aligned}
	\mathcal{H}_T &\leq  V(x_0,0,\mathcal{F}_0)-E \big \{V(x(T),T,\mathcal{F}(s(T)))  \big \} \\
	&+E \bigg\{\int_0^T \mathcal{X}^\top(v) \hat \varXi_k \mathcal{X}(v) dv \bigg\}
	\end{aligned}
	\end{equation} 
	Combining the LMIs~\eqref{LMI_h-infty} and $0 \leq E  \{V(x(T),T,\mathcal{F}(s(T))) \}$, inequality~\eqref{H_T} implies that
	\begin{equation}
	\mathcal{H}_{\infty} = \mathcal{H}_{T \to \infty} \leq V(x_0,0,\mathcal{F}_0).  
	\end{equation}
\end{proof}

We are now ready to give the proof of Theorem~\ref{thm:main:control}.

\begin{proof}[Proof of Theorem~\ref{thm:main:control}]
	The LMIs~\eqref{varXi} in Proposition~\ref{H_2 controller theorem} are equivalent to
	\begin{equation}\label{equivalent_39}
	\begin{aligned} 
	\bar \varXi_{k_{ij}} < 0
	\end{aligned}
	\end{equation}	
	for all $k_{ij}$ as shown in the proof lines of Proposition~\ref{H_2 controller theorem}. Using $Y_j =P^{-1}_j $ and $Y_j$ to pre-multiply and to post-multiply  $\bar \varXi_{k_{ij}}$, respectively, we obtain the following inequalities  
	\begin{equation} \label{Y_kmultiply}
	\begin{aligned}
	&Y_j A_i^\top  + A_i Y_j + (B_i Z_j Y^{-1}_j) \hat Q^{-1}_{k_{ij}} (Y^{-\top}_j Z^\top_j B^\top_i) \\ &+Y_j L^{-1}_{k_{ij}} Y_j + Y_j C^\top_i C_i Y_j + Y_j \bigg( \sum_{k'=1}^{N^2}  \kappa_{{k_{ij}}k'}P_{k'} \bigg) Y_j   < 0,
	\end{aligned}
	\end{equation}	
	where $L_{k_{ij}} = Q_{k_{ij}}^{-1}$ for each ${k_{ij}}\in \{1,\dotsc,N^2\}$. The inequalities~\eqref{Y_kmultiply} are equal to 
	\begin{equation} \label{h_2 ineq}
	\begin{aligned}
	\centering{ \begin{bmatrix}
		\mathcal{U}_{k_{ij}} & \star & \star  & \star   & \star  \\
		Y_j & -L_{k_{ij}} & \star  & \star  & \star  \\
		C_iY_j & 0 & -I_l & \star  & \star  \\
		Y_j^{-\top} Z_j^\top B_i^\top   & 0 & 0 & -\hat Q_{k_{ij}} & \star  \\
		\aleph_{k_{ij}}^\top & 0 & 0 & 0 & -\mathcal{Y}_{k_{ij}} \\
		\end{bmatrix}}  <0,  
	\end{aligned}
	\end{equation}
	from which we obtain
    \begin{flalign}
    & \begin{bmatrix} 
	\mathcal{U}_{k_{ij}} & \star  & \star  & \star  & \star  \\
	Y_j & -L_{k_{ij}} & \star  & \star  & \star  \\
	C_i Y_j & 0 & -I_l & \star  & \star  \\
	Z_j^\top B_i^\top & 0 & 0 & I_n & \star  \\
	\aleph_{k_{ij}}^\top & 0 & 0 & 0 & -\mathcal{Y}_{k_{ij}} \\
	\end{bmatrix} \\ 
	& + \begin{bmatrix}
	0 & \star  & \star  \\
	0 & (- Y_j^\top \hat Q_{k_{ij}} Y_j - I_n) & \star  \\
	0 &  0 & 0  \\
	\end{bmatrix} < 0, 
	\end{flalign}
 is the condition ensuring that the controller~\eqref{controller} is a $H_2$ controller. 
	Since $\hat Q_{k_{ij}} = (1-\tau^{+}) Q_{k_{ij}} \geq  0$, we have   
	\begin{equation} 
	\begin{aligned}
	\centering{\begin{matrix} 
		\begin{bmatrix}
		0 &  \star  & \star  \\
		0 & (- Y_j^\top \hat Q_{k_{ij}} Y_j - I) & \star  \\
		0 &  0 & 0  \\
		\end{bmatrix} 
		\end{matrix}}  < 0.  
	\end{aligned}
	\end{equation}
	Therefore, if the second LMIs of~\eqref{combineLMI} are satisfied for all $(i,j) \in \Theta \times \Theta$, then~\eqref{controller} is a $H_2$ controller of the system~$\Sigma_K$. Using the same argument in the above part of this proof, it can be proven that the third LMIs of~\eqref{combineLMI} are proposed for a $H_{\infty}$ controller design as presented in Proposition~\ref{H_f controller theorem}.
	 
	We then show the first condition of~\eqref{combineLMI}. Let $\lambda \in \mathbb{R}$ satisfy 
	\begin{equation}
	\begin{aligned}
	\phi^\top(0) P_{\mathcal{F}_0} \phi(0) &\leq \max_{k \in \{1,\dotsc,N^2\}} \bigg \{\phi^\top(0) P_k \phi(0) \bigg \} \leq \lambda, 
	\end{aligned}
	\end{equation}
	and $Y^{-1}_j = P_j$ for each $j \in \Theta$. Thus,
	\begin{equation}
	-\lambda + \phi^\top(0) Y^{-1}_j \phi(0) \leq 0, \phi(0) = x_0. 
	\end{equation}
    Moreover, 
	\begin{equation}
	\begin{aligned}
	\sup_{\mathcal{F}_0 \in \Theta \times \Theta} &\int_{-\tau}^0 x^\top(v) Q_{\mathcal{F}_0} x(v)dv \\ 
	&= \lambda_{\max}(Q_{\mathcal{F}_0}) \bigg[ \int_{-\tau}^0   x^\top(v) x(v)  dv  \bigg]
	\\ &= X^\top \lambda_{\max}(Q_{\mathcal{F}_0}) X.
	\end{aligned} 
	\end{equation}
	Let $\Lambda$ be a real number  satisfying
	\begin{equation}
	0< X^\top  \lambda_{\max}(L^{-1}_k) X  \leq  \Lambda,
	\end{equation}
    then
	\begin{equation}
	-\Lambda + X^\top \lambda_{\max}(L^{-1}_k) X \leq 0. 
	\end{equation}
	Therefore,
	\begin{equation}
	 x^\top_0P_{\mathcal{F}_0}x_0 + \int_{-\tau}^0 x^\top(s) Q_{\mathcal{F}_0}x(s)ds  \leq \lambda + \Lambda,
	\end{equation}
by which we obtain 
	\begin{equation}
	\mathcal{H}_2 \leq  \lambda + \Lambda, \mbox{and} \  \mathcal{H}_{\infty} \leq  \lambda + \Lambda 
	\end{equation}
    considering~\eqref{h_2_bound} and~\eqref{h_f_boudn}. If the real numbers $\lambda$ and $\Lambda$ satisfy the condition
   \[
\lambda + \Lambda \leq \min\{f_2,f_{\infty}\},
\]
then the properties in~\eqref{h_2 and h_inf measure} are satisfied. 
\end{proof}

\begin{figure}[tb] 
	\centering
	\includegraphics[width=70mm]{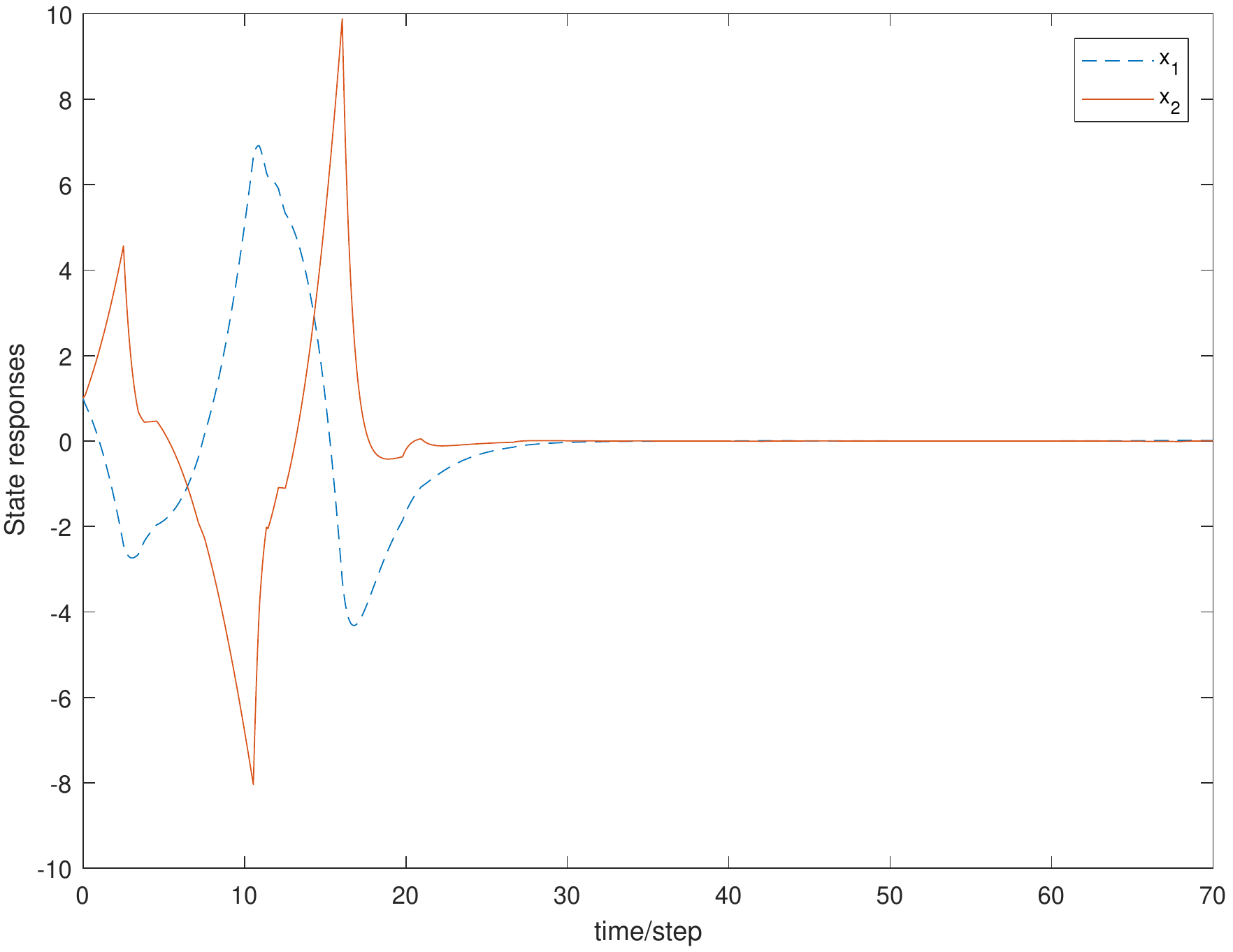}
	\caption{State trajectories \label{theorem4_x}}
\end{figure}
\begin{figure}[tb] 
	\centering
	\includegraphics[width=70mm]{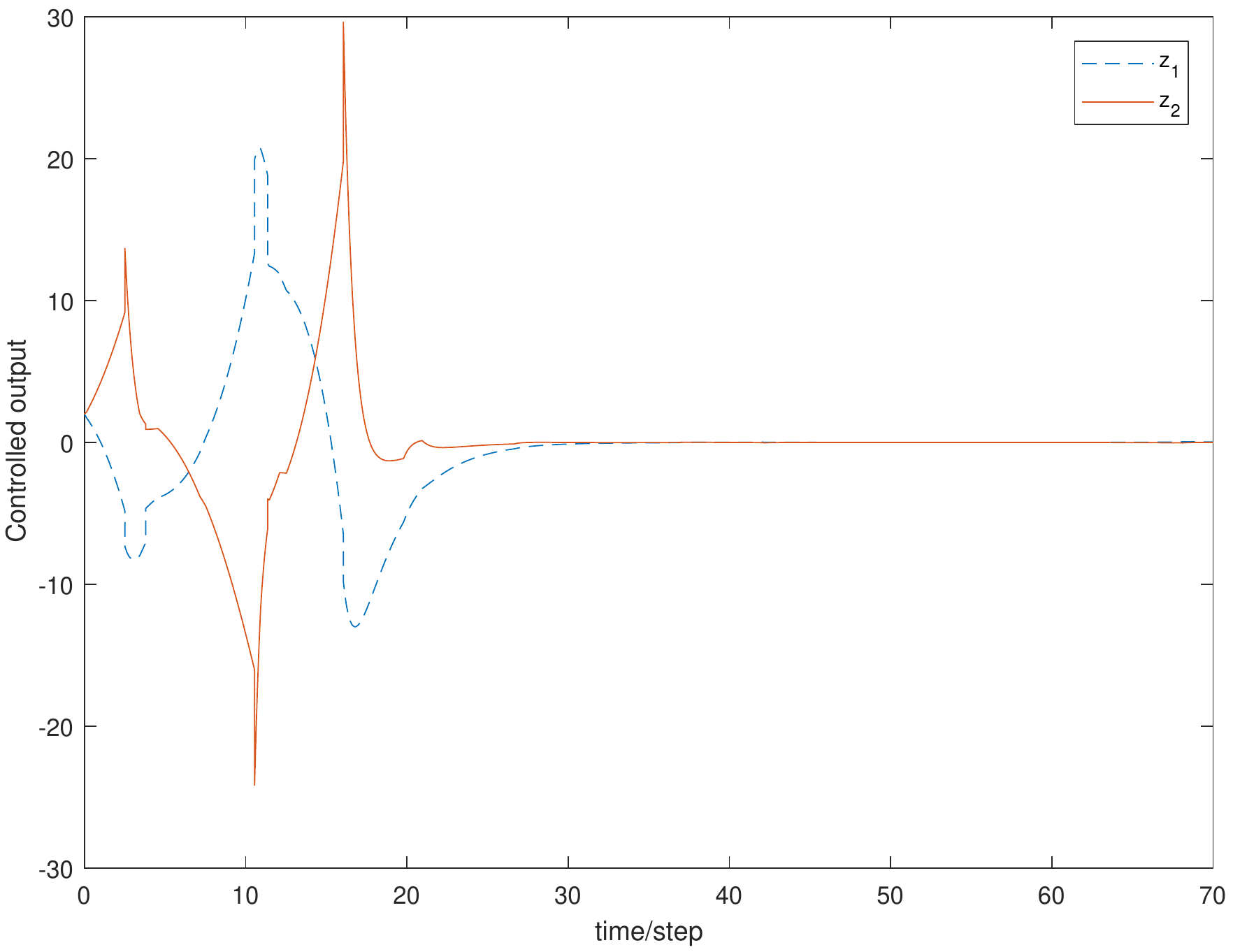}
	\caption{Controlled output $z(t)$ \label{theorem4_z}}
\end{figure}
\begin{figure}[tb] 
	\centering
	\includegraphics[width=70mm]{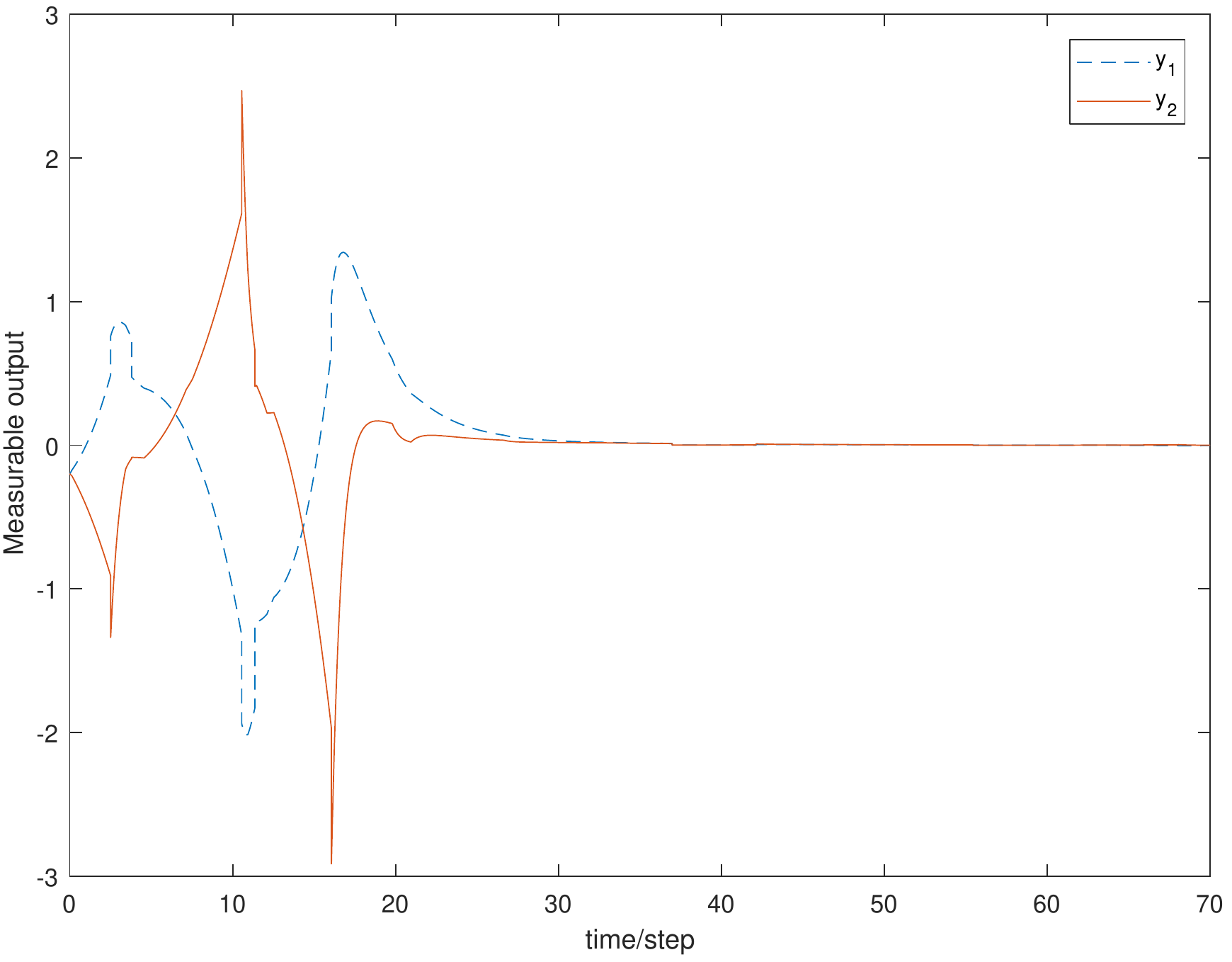}
	\caption{Measurable output $y(t)$ \label{theorem4_y}}
\end{figure}  
\begin{figure}[tb] 
	\centering
	\includegraphics[width=70mm]{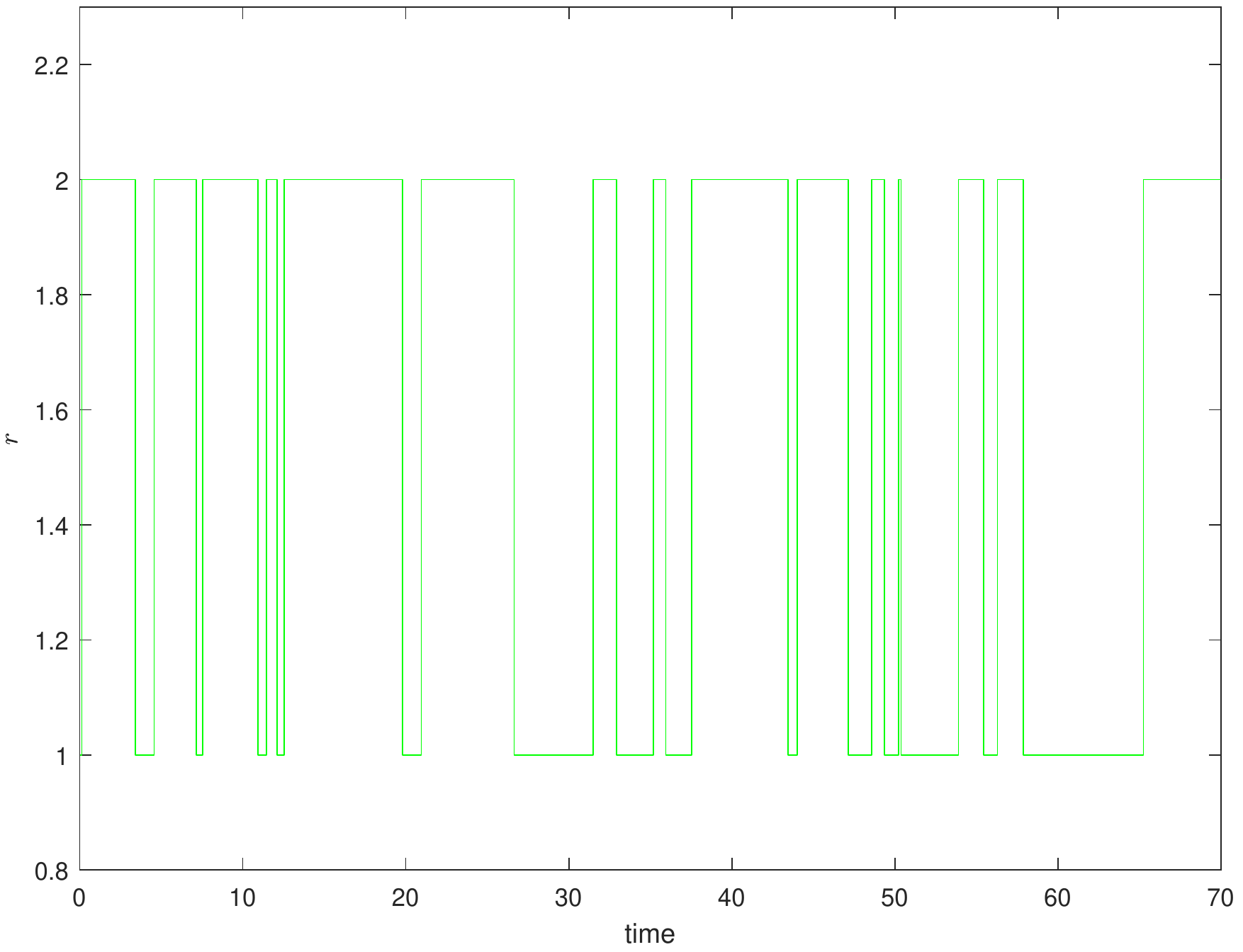}
	\caption{Mode signal $r(t)$ \label{r}}
\end{figure}
\begin{figure}[tb] 
	\centering
	\includegraphics[width=70mm]{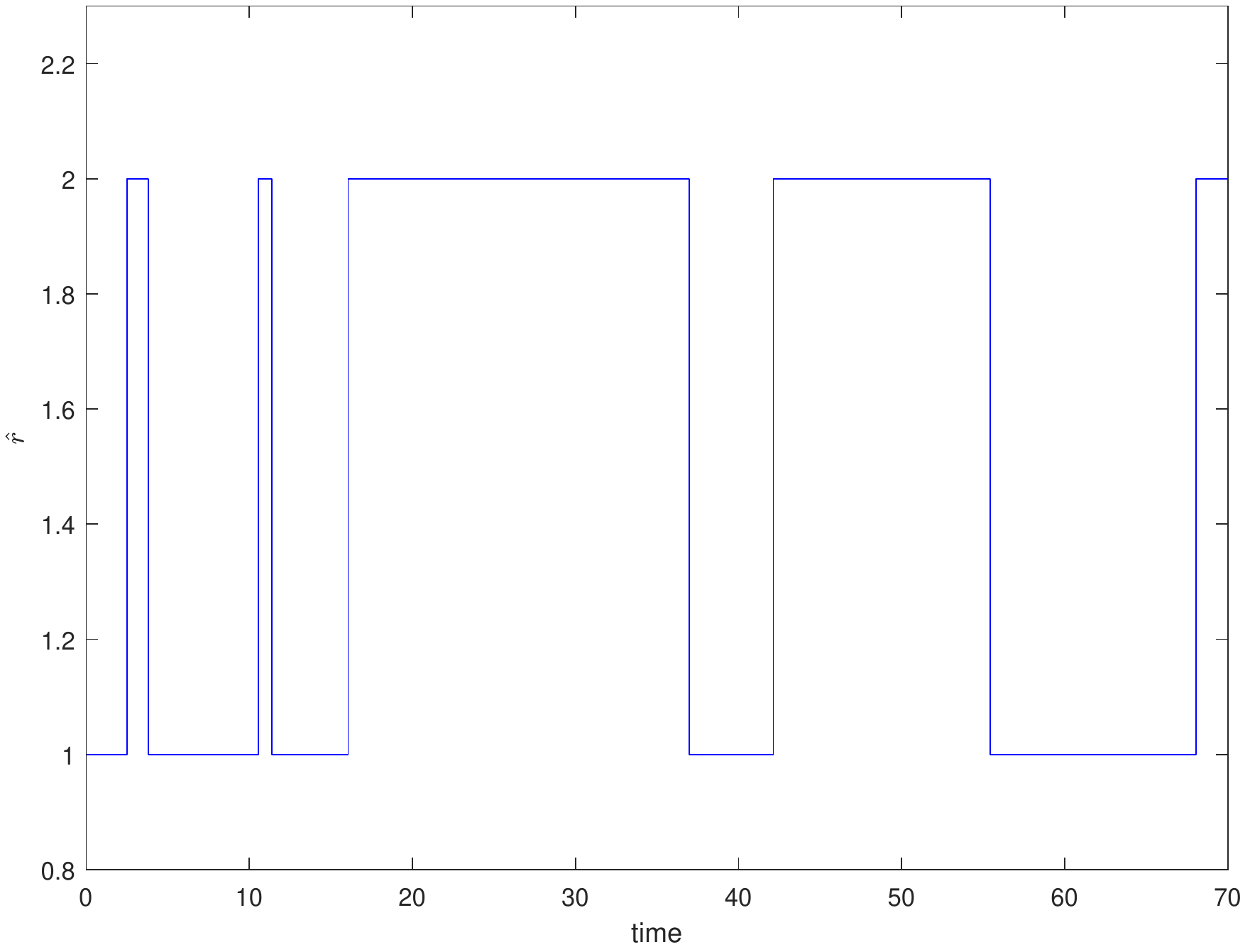}
	\caption{Mode observation $\tilde r(t)$ \label{hat_r}}
\end{figure}

\section{Numerical example} \label{numerical_example}

In this section, an example of mixed $H_2/H_{\infty}$ control is provided.
We let $\Theta = \{1,2\}$ in this example. 
We use $h_{1,2}$ to represent the observation delay of $\tilde r$ from the state 1 to the state 2, and $h_{2,1}$ to indicate the observation delay from the state 2 to the state 1. We assume that both $h_{1,2}$ and $h_{2,1}$ follow the exponential distribution with rate value $3$. Let the infinitesimal generator of  $r$ be  
\begin{equation}
\centering {\begin{matrix}
	\begin{bmatrix}
	-5&5   \\
	3&-3
	\end{bmatrix},
	\end{matrix}}
\end{equation}
so that 
\begin{equation}
\tilde{\mathcal{S}}= 
\begin{bmatrix}
-5&0&5&0   \\
3&-8&0&5   \\
3&0&-6&3    \\ 
0&3&0&-3
\end{bmatrix}.
\end{equation}
The system matrices are given as follows: 
\[
\centering{\begin{matrix}
	A_1 = \begin{bmatrix}
	0 &-0.45   \\
	0.9&0.9
	\end{bmatrix}, 	A_2 = \begin{bmatrix}
	0&-0.29   \\
	0.9&-1.26
	\end{bmatrix}, 	B_1 = \begin{bmatrix}
	0.5  \\
	1.1
	\end{bmatrix},
	\end{matrix}}  
\]
\[
\centering{\begin{matrix}
	B_2 = \begin{bmatrix}
		0.6  \\
		1.4
	\end{bmatrix}, E_1 = \begin{bmatrix}
		0&-0.01 \\
		-0.01 & -0.03
	\end{bmatrix},  E_2 = \begin{bmatrix}
		-0.01 & -0.03  \\
		-0.06 & -0.1
	\end{bmatrix},
\end{matrix}}
\]
\[
\centering{\begin{matrix}
	C_1 = \begin{bmatrix}
	2&0  \\
	0&2
	\end{bmatrix}, C_2 = \begin{bmatrix}
	3&0   \\
	0&3
	\end{bmatrix}, \varPsi_1 = \begin{bmatrix}
	0.4&0.5   \\
	-0.3&1.2
	\end{bmatrix},
	\end{matrix}}
\]
\[
\centering{\begin{matrix}
\varPsi_2 = \begin{bmatrix}
	-0.2&-0.4   \\
	0&-0.6
\end{bmatrix}, J_1 = \begin{bmatrix}
	-0.2 & 0 \\ 
	0 & -0.2
\end{bmatrix},
	\end{matrix}}
\]
\[
J_2 = \begin{bmatrix}
-0.3 & 0 \\ 
0 & -0.3
\end{bmatrix},
\varPhi_1  = \begin{bmatrix}
0.1&0   \\
0&0.1
\end{bmatrix}, \varPhi_2 = \begin{bmatrix}
1&0   \\
0&1
\end{bmatrix}.
\]
Regarding the mixed $H_2/H_{\infty}$ control, we let \[\displaystyle \tau^{+} = 0.5, \gamma = 1, f_2 = 15, f_{\infty} = 17,
\] 
\[Q_1 = Q_2 = Q_3 = Q_4 = I_{2},\] \[w(t) = [0.5e^{-0.1t} \sin(0.01 \pi t) \quad  0.5e^{-0.1t} \sin(0.01 \pi t)]^\top,\]
so that $\hat Q_k = (1-\tau^{+})Q_k = 0.5 I_2$ for all $k \in \{1,\dotsc,4\}$. Also, we set the initial state $\phi(0) = [0 \ \ 1]^\top$ and $X =  2$. The simulation result shows that there exist 
\[
\centering{\begin{matrix}
	Y_1 = \begin{bmatrix}
	0.1109 &	-0.0541 \\
	-0.0541 &  0.2465
	\end{bmatrix}, Y_2 = \begin{bmatrix}
	0.0702	& -0.0433 \\
	-0.0433 & 0.1729
	\end{bmatrix}, 
	\end{matrix}} 
\]
\[
Z_1 = [-0.0603  \ \ -0.0603]^\top, Z_2 = [-0.0209  \ \ -0.0209]^\top, 
\]
\[
\lambda = 7.1444, \Lambda = 4 
\] 
such that the LMIs of~\eqref{combineLMI} are established. Thus, we have
\begin{equation}
\centering{\begin{matrix}
	K_1 = \begin{bmatrix}
	-0.7423 & -0.4074
	\end{bmatrix}, K_2 = \begin{bmatrix}
	-0.4397 &  -0.2309
	\end{bmatrix}.
	\end{matrix}}
\end{equation}
Finally, we present the curve of state trajectories, controlled output, measurable output, mode signal, and mode observation of the system~$\Sigma_K$ versus $t$ in Fig~\ref{theorem4_x}, Fig~\ref{theorem4_z}, Fig~\ref{theorem4_y}, Fig~\ref{r}, and Fig~\ref{hat_r}, respectively.

\section{Conclusion}
In this work, for the continuous-time Markov jump linear systems with state and mode-observation delays, we have proposed a generic scheme for studying and devising state feedback control methods. Specifically, the delay of mode-observation is hypothesized to follow an exponential distribution. We have investigated that it is possible to remodel the closed-loop system as a Markov jump linear system with state delay in a standard form. On the basis of the remodeling, we devise a mixed $H_2/H_{\infty}$ controller by using an LMI framework. We have also examined the effectiveness of our proposed result by an example. Our further direction of research will contain reformulation of the  LMIs  to better  take the
features of  the considered  problems into account. It is also interesting to consider using other forms of Lyapunov function for studying the problems, e.g., a Lyapunov function with a three-part form. The third potential direction is an analysis of sliding mode control of the system in this paper.

\section*{References}
\begin{enumerate}	
	\item[{[1]}] Farias, D., Geromel, J., Do Val, J., Costa, O.:'Output feedback
  control of Markov jump linear systems in continuous-time', IEEE
  Transactions on Automatic Control, 2000, 45, (2),  pp. 944--949\vspace*{6pt}
   
   \item[{[2]}] Mahmoud, S., Shi, P.:'Robust stability, stabilization and $H_{\infty}$ control
  of time-delay systems with Markovian jump parameters', International Journal of Robust and Nonlinear Control, 2003, 784, (8), pp. 755--784\vspace*{6pt}
	
	\item[{[3]}] Xu, S., Lam, J., Mao, X.:'Delay-dependent $H_{\infty}$ control and filtering for
  uncertain Markovian jump systems with time-varying delays', IEEE
  Transactions on Circuits and Systems I: Regular Papers, 2007, 54, (9), pp.
  2070--2077\vspace*{6pt}
  
  \item[{[4]}] Zhao, X., Zeng, Q.:'New robust delay-dependent stability and $H_{\infty}$ analysis for uncertain Markovian jump systems with time-varying delays', Journal of the Franklin Institute, 2010, 347, (5), pp. 863--874\vspace*{6pt}

  \item[{[5]}] Mhaskar, P., El-Farra, N., Christofides, P.:'Robust predictive
  control of switched systems: Satisfying uncertain schedules subject to state
  and control constraints', International Journal of Adaptive Control
  and Signal Processing, 2008, 22, (2), pp. 161--179\vspace*{6pt}
  
  \item[{[6]}] Shi, P., Boukas, E., Liu, Z.:'Delay-dependent stability and output feedback
  stabilisation of Markov jump system with time-delay', IEE Proceedings
  - Control Theory and Applications, 2002, 149, (5), pp. 379--386\vspace*{6pt}
  
   \item[{[7]}] Cao, Y., Lam, J.:'Stochastic stabilizability and $H_{\infty}$ control for
  discrete-time jump linear systems with time delay', Journal of the
  Franklin Institute, 1999, 336, (8), pp. 1263--1281\vspace*{6pt} 
  
   \item[{[8]}] Cao, Y., Lam, J.:'Robust $H_{\infty}$ control of uncertain Markovian jump systems with time-delay', IEEE Transactions on Automatic Control, 2000, 45, (1), pp. 77--83\vspace*{6pt}
  
  \item[{[9]}] Chen, W., Guan, Z., Yu, P.:'Delay-dependent stability and $H_{\infty}$ control of
  uncertain discrete-time Markovian jump systems with mode-dependent time
  delays', Systems {\&} Control Letters, 2004, 52, (5), pp. 361--376\vspace*{6pt}
  
  \item[{[10]}] Xiong, J., Lam, J.:'Stabilization of discrete-time Markovian jump linear
  systems via time-delayed controllers', Automatica, 2006, 42, (5), pp. 747--753\vspace*{6pt}
  
  \item[{[11]}] Cetinkaya, A., Hayakawa, T.:'Discrete-time switched stochastic control
  systems with randomly observed operation mode', 52nd IEEE
  Conference on Decision and Control, Florence, Italy, Dec 2013, pp. 85--90.  \vspace*{6pt} 
  
  \item[{[12]}] Cetinkaya, A., Hayakawa, T.:'Stabilizing discrete-time switched linear stochastic systems using
  periodically available imprecise mode information', 2013 American
  Control Conference, Washington, USA, Jun 2013, pp. 3266--3271. \vspace*{6pt}
  
  \item[{[13]}] Cetinkaya, A., Hayakawa, T.:'Sampled-mode-dependent time-varying control strategy for stabilizing discrete-time switched stochastic systems', 2014 American Control Conference, Portland, USA, Jun 2014, pp.
  3966--3971. \vspace*{6pt}
  
   \item[{[14]}] Lou, X., Cui, B.:'Delay-dependent stochastic stability of delayed Hopfield neural networks with Markovian jump parameters', Journal of Mathematical Analysis and Applications, 2007, 328, (1), pp. 316--326. \vspace*{6pt}
   
   \item[{[15]}]  Chen, W., Zheng, W., Shen, Y.:'Delay-dependent stochastic stability and $H_{\infty}$-control of uncertain neutral stochastic systems With time delay', IEEE Transactions on Automatic Control, 2009, 54, (7), pp. 1660--1667. \vspace*{6pt}
   
   \item[{[16]}] Sakthivel, R., Harshavarthini, S., Kavikumar, R., Ma, Y.:'Robust tracking control for fuzzy Markovian jump systems with time-varying delay and disturbances', IEEE Access, 2018, 6, pp. 66861--66869.  \vspace*{6pt}
  
  \item[{[17]}] Boukas, E., Liu, Z., Shi, P.:'Delay-dependent stability and output feedback stabilisation of Markov jump system with time-delay', IEE Proceedings - Control Theory and Applications, 2002, 149, (5), pp. 379--386.  \vspace*{6pt}
  
  \item[{[18]}] Hien, L., Trinh, H.:'Delay-dependent stability and stabilisation of two-dimensional positive Markov jump systems with delays', IET Control Theory and Applications, 2017, 11, (10), pp. 1603--1610.  \vspace*{6pt}
  
  \item[{[19]}] Sakthivela, R., Sakthivel, R., Nithyaa, V., Selvaraj, P., Kwon, M.:'Fuzzy sliding mode control design of Markovian jump systems with time-varying delay', Journal of the Franklin Institute, 2018, 335, (14), pp. 6353--6370. \vspace*{6pt}
  
  \item[{[20]}] Park, B., Kwon, N., Park, P.:'Stabilization of Markovian jump systems with incomplete knowledge of transition probabilities and input quantization', Journal of the Franklin Institute, 2015, 352, (10), pp. 4354--4365.  \vspace*{6pt}
   
  \item[{[21]}] Xie, X., Lam, J., Fan, C.:'Robust time-weighted guaranteed cost control of uncertain periodic piecewise linear systems', Information Sciences, 2018, 460, pp. 238--253.  \vspace*{6pt}
  
  \item[{[22]}] Chen, B., Liu, P.:'Delay-dependent $H_2/H_{\infty}$ control for a class of switched TS fuzzy systems with time-delay',  IEEE Transactions on Fuzzy Systems, 2005, 13, (4), 544--556.  \vspace*{6pt}
  
  \item[{[23]}]  Aliyu, M., Boukas, E.:'Mixed $H_2/H_{\infty}$ stochastic control problem',  IFAC Proceedings Volumes, 1999, 32, (2), pp. 4929--4934. \vspace*{0pt}

  \item[{[24]}] Boukas, E.:'$H_{\infty}$ control of discrete-time Markov jump systems with bounded transition probabilities', Optimal Control Applications and Methods, 2009, 30, (5), pp. 477--494.  \vspace*{6pt}

  \item[{[25]}] Luan, X., Zhao, S., Liu, F.:'$H_{\infty}$ control for discrete-time Markov jump systems with uncertain transition probabilities', IEEE Transactions on Automatic Control, 2012, 58, (6), pp. 1566--1572.  \vspace*{6pt}

  \item[{[26]}] Ma, S., Zhang, C.:'$H_{\infty}$ control for discrete-time singular Markov jump systems subject to actuator saturation', Journal of the Franklin Institute, 2012, 349, (3), pp. 1011--1029.  \vspace*{6pt}
   
  \item[{[27]}] Feng, X., Loparo, K., Ji, Y., Chizeck, H.:'Stochastic stability properties
  of jump linear systems', IEEE Transactions on Automatic Control, 1992, 37, (1), pp. 38--53\vspace*{6pt}
  
  \item[{[28]}] Mahmoud, M., Al-Muthairi, N.:'Design of robust controllers for time-delay
  systems', IEEE Transactions on Automatic Control, 1994, 39, (5), pp. 995--999\vspace*{6pt}
  
  \item[{[29]}] Mahmoud, M., AL-Sunni, F., Shi, Y.:'Mixed control of uncertain
  jumping time-delay systems', Journal of the Franklin Institute, 2008, 345, (5), pp. 536--552 \vspace*{6pt}
\end{enumerate}

\end{document}